\documentclass[runningheads,letter]{llncs}
\usepackage{graphicx}
\usepackage[utf8]{inputenc}
\usepackage[english]{babel}

\usepackage{amsmath}
\usepackage{amssymb}
\usepackage{mathtools}
\usepackage{csquotes}

\usepackage{csquotes}
\usepackage[%
	backend=biber, style=numeric-comp, sorting=nyt, natbib=true, maxbibnames=99, giveninits=true,
    sortcites=true,doi=false,url=false, isbn=false,eprint=false]{biblatex}

\usepackage{arydshln}
\usepackage{cleveref}

\newtheorem{innercustomgeneric}{\customgenericname}
\providecommand{\customgenericname}{}
\newcommand{\newcustomtheorem}[2]{%
  \newenvironment{#1}[1]
  {%
   \renewcommand\customgenericname{#2}%
   \renewcommand\theinnercustomgeneric{##1}%
   \innercustomgeneric
  }
  {\endinnercustomgeneric}
}

\newcustomtheorem{customlemma}{Lemma}

\newtheorem{observation}{Observation}

\makeatletter
\renewcommand{\fnum@figure}{Figure \thefigure}
\makeatother

\title{Complexity of the Multiobjective Spanner Problem}

\begin{document}
\title{Complexity of the Multiobjective Spanner Problem}
%
%
\author{Fritz Bökler\orcidID{0000-0002-7950-6965} \and
        Henning Jasper\orcidID{0000-0002-9821-8600}
}
\authorrunning{Fritz Bökler and Henning Jasper}
%
\institute{Universität Osnabrück, Osnabrück Niedersachen 49074, Germany\\
\email{\{fboekler, hejasper\}@uos.de}}
\maketitle              
\begin{abstract}
In this paper, we take an in-depth look at the complexity of a hitherto unexplored \emph{Multiobjective Spanner (MSp)} problem. The MSp is a multiobjective generalization of the well-studied Minimum t-Spanner problem. This multiobjective approach allows us to find solutions that offer a viable compromise between cost and utility. Thus, the MSp can be a powerful modeling tool when it comes to the planning of, e.g., infrastructure. We show that for degree-3 bounded outerplanar instances the MSp is intractable and computing the non-dominated set is \textbf{BUCO}-hard. Additionally, we prove that if $\mathbf P \neq \mathbf {NP}$, neither the non-dominated set nor the set of extreme points can be computed in output-polynomial time, for instances with unit costs and arbitrary graphs. Furthermore, we consider the directed versions of the cases above.

\end{abstract}
\section{Introduction}
The \emph{Multiobjective Spanner (MSp)} problem is a multiobjective generalization of the \emph{Minimum t-Spanner} problem.
Given a connected, simple graph $G=(V,E)$ where every edge has a cost and length of $1$, a subset of edges $S$ is a \emph{t-spanner} of $G$ if for every pair of vertices $u,v \in V$, $\frac{d^S(u,v)}{d^E(u,v)} \leq t$ holds, with $d^S(u,v)$ being the distance from $u$ to $v$ in $S$ and $d^E(u,v)$ their respective distance in $E$ \cite{peleg1989}. For a given graph, the problem of finding the cheapest t-Spanner, with regard to the sum over all edge-costs, is commonly known as the Minimum t-Spanner problem. We refer to $t$ as the \emph{stretch factor}.
The MSp generalizes the Minimum t-Spanner problem by introducing two edge-weight functions, allowing us to assign each edge a cost independent of its length. Furthermore, in contrast to the Minimum t-Spanner problem, the goal of the MSp is not to find a minimum weight spanner for a given stretch factor. Instead, the stretch factor is another objective we aim to minimize. The stretch factor is an interesting objective function in itself that is to be minimized in, e.g., the \emph{Minimum Max-Stretch Spanning Tree (MMST)} problem \cite{corneil1995}.
Feasible solutions for the MSp are defined less restrictive than t-Spanners. We define a \emph{spanner} of a connected, undirected graph $G=(V,E)$, as a subset of edges $S \subseteq E$, such that $G'=(V,S)$ is a connected subgraph.
Since the MSp is a \emph{Multiobjective Combinatorial Optimization (MOCO)} problem, solutions are mapped to a \emph{value vector} instead of a single value. Due to conflicts among objectives, there does not necessarily exist a solution that achieves the best value in all objective functions simultaneously. Instead, we look for value vectors for which there are no other value vectors that dominate them. This set of value vectors is called \emph{non-dominated set} (or Pareto-front) $\mathcal{Y}_N$.

\begin{definition}[Multiobjective Spanner (MSp) Problem]
The input is a connected, undirected graph $G=(V,E)$ and edge-weight functions $c_1\colon E \rightarrow \mathbb{Z}$ and $c_2\colon E \rightarrow \mathbb{N}_+$. Feasible solutions are spanners $S$ of $G$ and are assessed based on the two objective functions
\[ f_1(S)= \sum_{e \in S} c_1(e) \text{ and } f_2(S) = \max_{u,v \in V} \frac{d^S_{c_2}(u,v)}{d^E_{c_2}(u,v)},\]
with $d^S_{c_2}(u,v)$ and $d^E_{c_2}(u,v)$ being the length of the shortest u-v-path in $S$ and $E$ respectively, regarding the $c_2$-length. We consider an instance of the MSp to be solved if we output its non-dominated set $\mathcal{Y}_N$.
\end{definition}
A possible field of application for the MSp is the planning of emergency infrastructure.
Natural disasters require a significant logistical effort to provide relief to victims and to distribute equipment and humanitarian goods. The efficient design of emergency infrastructure therefore forms the basis for initial responses, as well as for long-term measures taken to stabilise affected communities.
Many optimization models used in emergency logistics only focus on either cost-effectiveness or responsiveness \cite{chen2020, boonmee2017, ozdamar2004}, as multiobjective approaches are considered too computationally expensive to solve. A recent literature review by Caunhye et al. \cite{caunhye2012} concluded, that these singleobjective models may hamper relief services by causing an oversupply of resources leading to difficulty with coordination, greater traffic, and complex scheduling. Therefore, the MSp could help to address this shortcoming.

The concept of t-Spanners and their related problems were first introduced by Peleg, Schäffer and Ullman in the context of synchronization in distributed systems and communication networks \cite{peleg1989, Ullman1989} and have since been explored in a variety of publications. 
A greedy algorithm with one of the best cost-guarantees was developed by Althöfer \cite{althofer1993}. For a graph $G$ and a stretch factor of $t=2k-1 (k \in \mathbb{N}_{\geq 1})$, it creates a t-Spanner $S$ of $G$ containing $\mathcal{O}(n^{1+1/k})$ edges in $\mathcal{O}(m(n^{1+1/k}+n \log n))$. This algorithm can even be applied to the \emph{Weighted Minimum t-Spanner} problem, where every edge is assigned an arbitrary positive cost. Then the algorithm additionally guarantees that $S$ has a cost of at most $\mathcal{O}(n/k)$ times the cost of the minimum spanning tree of $G$.
For undirected Minimum 2-Spanners, Kortsarz and Peleg \cite{kortsarz1994} published a $\mathcal{O}(\log(m/n))$-approximation with a theoretical running time of $\mathcal{O}(m^2 n^2 \log(n^2 /m))$.
Baswana and Sen \cite{baswana2007} give a method that, for a weighted, undirected graph, computes a $t=2k-1$ spanner $S$ that contains at most $\mathcal{O}(kn^{1+1/k})$ edges in an expected running time of $\mathcal{O}(km)$, but with no cost guarantee. For unweighted graphs the size of $S$ is bounded by $\mathcal{O}(n^{1+1/k}+kn)$.
Cai and Keil \cite{CaiKeil1994} focused on the complexity of the Minimum t-Spanner problem for degree bounded graphs and showed, among others, that if the maximum degree of the graph is at most $4$, the Minimum 2-Spanner problem can be solved in linear time, whereas the problem is \textbf{NP}-hard even if the maximum degree is at most $9$.
A recent paper by Kobayashi \cite{kobayashi2018} focuses on the complexity of the Minimum t-Spanner problem in planar graphs and, as a byproduct, improves the degree bounds for \textbf{NP}-hardness found by Cai and Keil.

As many decisions require the consideration of multiple goals and conflicting demands, MOCO problems are an important modelling tool in a variety of fields. Practical applications include routing problems in public transport \cite{delling2015, wagner2017}, the planning of radiotherapy \cite{hamacher1999, thieke2007, giantsoudi2013} and the determination of control strategies for vaccine administration in COVID-19 pandemic treatment \cite{libotte2020}.
In the multiobjective context, a problem is called \emph{intractable}, if there is no algorithm capable of solving it in polynomial time \cite{ehrgott2005}.
Due to the exponential size of their non-dominated sets, many interesting MOCO problems are intractable, e.g., multiobjective variants of the Traveling Salesperson \cite{emelichev1992}, Shortest Path \cite{hansen1980} or Spanning Tree \cite{hamacher1994} problem. Therefore, it makes sense to consider a complexity class that distinguishes between problems that cannot be solved in polynomial time due to the size of their output and the ones that are genuinely hard to solve. Moreover, in experimental studies, the non-dominated sets are much smaller (e.g., \cite{BC20}). There is also a theoretical reason for this behavior: In a smoothed analysis setting, Brunsch and Röglin showed that the size of the non-dominated set is at most polynomial in the input size for each fixed number of objectives \cite{BR15}. 
We say a MOCO problem $O$ is solvable in \emph{output-polynomial time} if there is an algorithm that, for any given instance $I$ of $O$ outputs every $y \in \mathcal{Y}_N$ exactly once, in polynomial time depending on the size of the input $I$ and the output $\mathcal{Y}_N$ \cite{johnson1988}. Such an algorithm is called \emph{output-polynomial}. We denote the class of problems for which an output-polynomial algorithm exists as \textbf{OP}.

An interesting subset of the non-dominated set is the set of extreme points $\mathcal{Y}_X$.
Since making decisions based on a potentially exponentially sized non-dominated set is generally not practical, many MOCO problems are approached by combining all the objective functions into one singleobjective scalar (or preference) function. One method to accomplish this is called \emph{weighted sum scalarization (WSS)}, where each objective function is weighted according to its importance. Note that in general not all non-dominated points can be found in this way.
The extreme points of a MOCO problem instance are exactly the points that can be the solution to any WSS of the instance. Therefore, if every decision maker has a linear preference function, computing the extreme points suffices. Note, however, determining weights accurately reflecting the preferences of the deciders is not trivial.
As every extreme point is non-dominated, while not every non-dominated value vector is an extreme point, solving the MSp could be hard, while the problem of only computing the set of extreme points could be in \textbf{OP}. This is the case for, e.g., \emph{Multiobjective Shortest Path} \cite{Bokler2015, Bokler2017}.
For more information on MOCOs and related topics cf. the book by M. Ehrgott \cite{ehrgott2005}.

\vspace{-3pt}
\subsubsection{Contribution and Organisation.}
In the remainder of this paper, we first give some definitions and establish basic concepts and results in \Cref{section:Preliminaries}. 
In \Cref{section:Intractability}, we study the classic tractability of MSp.
\begin{theorem}\label{theorem:MSp_intractable}
MSp is intractable even on degree-3 bounded outerplanar graphs.
\end{theorem}
This is an interesting result, as there are non-trivial stretch factors for which the Minimum t-Spanner problem is solvable in linear time, under such restrictions.
In \Cref{section:Non-Dominated Set}, we first consider the output-sensitive complexity of computing the non-dominated set for unweighted instances of the MSp, where each edge has a cost and length of $1$.
\begin{theorem}\label{theorem:unweighted_MSp_OP}
If \textbf{P} $\neq$ \textbf{NP}, then MSp $\notin$ \textbf{OP}, even for unweighted instances.
\end{theorem}
Afterwards, we consider the \emph{BUCO} problem that can be interpreted as an unrestricted version of the Knapsack problem and discuss the output-sensitive complexity of computing the non-dominated set for degree bounded outerplanar instances of the MSp.
While BUCO appears to be a straight forward problem, it is currently unknown whether it can be solved in output-polynomial time \cite{BUCO_complexity2020}. However, it has been shown that there are other problems of unknown output-sensitive complexity that the BUCO problem can be reduced to. This motivated the introduction of the complexity-class of \textbf{BUCO}-hard problems \cite{boklerDIS}.
\begin{theorem}\label{theorem:MSp_BUCO_hard}
MSp is \textbf{BUCO}-hard even on degree-3 bounded outerplanar graphs.
\end{theorem}
Moreover, this theorem implies that if there is a polynomial time algorithm for the minimum $t$-spanner problem on degree-3 bounded outerplaner graphs where $t>1$ is part of the input then BUCO can be solved in output polynomial time.

As \Cref{theorem:unweighted_MSp_OP} states that we cannot compute the entire non-dominated set of unweighted MSp instances in output-polynomial time, in \Cref{section:Extreme Points} we define the problem of computing the set of extreme points for instances of the MSp  (MSp\textsuperscript{YEx}) and show its hardness with regard to output-sensitive complexity. 
\begin{theorem}\label{theorem:MSP_YEx_notin_OP}
If \textbf{P} $\neq$ \textbf{NP}, then MSp\textsuperscript{YEx} $\notin$ \textbf{OP}, even for unweighted instances.
\end{theorem}
Finally, \Cref{section:Conclusion} has concluding remarks. More details can be found in the appendix. Note that we also define a directed version of the MSp (diMSp), for which the same results are proven. The only exemption being \Cref{theorem:MSp_BUCO_hard}. The diMSp is \textbf{BUCO}-hard, even for degree-4 bounded outerplanar instances.

\section{Preliminaries}\label{section:Preliminaries}
We denote $\mathbb{N}=\{0,1,2,...\}$, $\mathbb{N}_+\coloneqq \mathbb{N} \setminus \{0\}$ and non-negative real numbers as $\mathbb{R}_{\geq}$.
For $n \in \mathbb{N}$, we denote the set $\{1,...,n\}$ as $[n]$.
For a graph $G=(V,E)$, an edge $\{u,v\} \subseteq E$ and an edge-weight function $c_i\colon E \rightarrow \mathbb{Z}$, $i\in [2]$ we abbreviate $c_i(\{u,v\})$ by $c_i(u,v)$. Furthermore, for a set of edges $S \subseteq E$, we denote $c_i(S)=\sum_{e \in S} c_i(e)$.
In order to simplify the input of the (di)MSp, we sometimes combine the edge-weight functions $c_1$ and $c_2$ into one function $c\colon E \rightarrow \mathbb{Z} \times \mathbb{N}_+$ with $c(e)=(c_1(e), c_2(e))^\mathsf{T}$. 
Similarly, for an instance of the (di)MSp and a feasible (directed) spanner $S$, we combine the two objective functions $f_1(S)$ and $f_2(S)$ into one single function $f(S)=(f_1(S),f_2(S))^\mathsf{T}$ that directly maps $S$ to its value vector.
We sometimes refer to $c\colon E \rightarrow \{ (1,1)^\mathsf{T} \}$ with $c(e)=(c_1(e), c_2(e))^\mathsf{T}=(1,1)^\mathsf{T}$ for all $e \in E$ as the \emph{trivial edge-weight function} and call instances of the MSp with these edge-weight functions \emph{unweighted}.

The degree of a vertex in an undirected graph is the number of vertices it is adjacent to.
The degree of a vertex in a directed graph is the number of its in- and out-going edges.
For $\delta \in \mathbb{N}$, we call any graph $G=(V,E)$ \emph{degree-$\delta$ bounded} if for all $v \in V$ their degree is less than or equal to $\delta$.
We call graphs \emph{outerplanar} if they have a drawing, in which every vertex lies on the boundary of the outer face.
We call undirected graphs \emph{connected} if they are non-empty and any two of their vertices are linked by a path.
A directed graph is called \emph{weakly connected} if replacing all of its arcs with undirected edges results in a connected (undirected) graph. See also \cite{diestel2017, bang2008}.
For an instance of a MOCO problem with an objective function $f$, we denote the set of all its value vectors as $\mathcal{Y}$. 
For unequal value vectors $y,y'\in \mathcal{Y}$, we say $y$ is dominated by $y'$ if $y'$ is component wise less than or equal to $y$.
Analogously, for feasible solutions $S, S'$ we say $S$ is dominated by $S'$ if $f(S)$ is dominated by $f(S')$.
If a value vector is not dominated by any value vector, the associated solution is called \emph{Pareto-optimal}.


For a weakly connected, directed graph  $G=(V,A)$, we call a subset of arcs $S\subseteq A$ a \emph{directed spanner} of $G$, if $G'=(V,S)$ is a subgraph such that, for every pair of vertices $u,v \in V$, if there is a directed u-v-path in $E$, there is one in $S$ as well.

Note, that the definition of a (directed) spanner does not require the resulting subgraph to be acyclic.
Analogously to the MSp, we define the \emph{Directed Multiobjective Spanner (diMSp)} problem. The only differences being that the input is a weakly connected, directed graph, solutions are now directed spanners and that in the second objective function, we only consider pairs of vertices that are connected in the initial graph. This guarantees the well-definedness of the objective function values.

\begin{lemma}\label{lemma:poly_restricted_of}
For a set of (di)MSp instances $\mathcal{I}$, if there is a polynomial $p\colon \mathbb{N}\rightarrow \mathbb{N}$, such that for every instance $I \in \mathcal{I}$ and its set of solutions $\mathcal{S}\colon |f_i(\mathcal{S})| \leq p(|I|)$ for $i=1$ or $i=2$, then $|\mathcal{Y}_N| \leq p(|I|)$.
\end{lemma}
\begin{proof}
Without loss of generality, we can assume that $f_1$ only has polynomially many different values in its image.
For every $a \in f_1(\mathcal{S})$, there is one $s' \in S$ with $f_1(s')=a$ and $f_2(s')\leq f_2(s)$ for all $s \in S$ with $f_1(s)=a$. Hence, $(a,f_2(s'))^\mathsf{T}$ dominates $(a,f_2(s))^\mathsf{T}$ for all $s \in S$ with $f_1(s)=a$. Thus, for each $a \in f_1(\mathcal{S})$ there is only one non-dominated value vector. \hfill \qed
\end{proof}
\begin{observation}\label{observation:spanner_adding_edges_with_c1=0}
It is clear that adding edges to a spanner never increases its stretch factor and that therefore, for every non-dominated value vector $y \in \mathcal{Y}_N$ there is a spanner $S$ with $f(S)=y$ and $e \in S$ for all $e \in E$ with $c_1(e) = 0$.
\end{observation}
We call the decision problem corresponding to the Minimum t-Spanner problem t-Spanner\textsuperscript{DEC}. For it, verifying the stretch factor $t$ of a spanner only requires considering pairs of vertices that are connected in the underlying graph \cite{peleg1989}. In case of the MSp an analogous statement can be made.

\begin{lemma}\label{lemma:check_t-spanner}
Let $I=(G=(V,E),c_1,c_2)$ be a MSp instance with a connected, undirected graph $G$ and edge-weight functions $c_1$ and $c_2$. 
For any spanner $S$ of $G$, $f_2(S) =\max_{u,v \in V} \frac{d^S_{c_2}(u,v)}{d^E_{c_2}(u,v)}=\max_{\{u,v\}\in E} \frac{d^S_{c_2}(u,v)}{d^E_{c_2}(u,v)}$ holds.
\end{lemma}
\begin{proof}
Let $S$ be a spanner of $G$ and assume $f_2(S)=\frac{d^S_{c_2}(r,z)}{d^E_{c_2}(r,z)}$ with $\{r,z\} \notin E$. Let $\{r=u_0,u_1\}, \{u_1,u_2\},...,\{u_{m-1},u_m=z\}$ be the shortest r-z-path in $E$. Denote the set of pairs of vertices $(u_i, u_{i+1})$, $0\leq i \leq m-1$ as $U$.
We get
\begin{align*}
\frac{d^{S}_{c_2}(r,z)}{d^{E}_{c_2}(r,z)} & \leq \frac{\sum_{i=0}^{m-1} d^{S}_{c_2}(u_i, u_{i+1})}{ \sum_{i=0}^{m-1} d^{E}_{c_2}(u_i, u_{i+1})}
 \leq \max_{(u_i, u_{i+1}) \in U} \frac{d^{S}_{c_2}(u_i, u_{i+1})}{d^{E}_{c_2}(u_i, u_{i+1})} \cdot \frac{\sum_{i=0}^{m-1} d^{E}_{c_2}(u_i, u_{i+1})}{ \sum_{i=0}^{m-1} d^{E}_{c_2}(u_i, u_{i+1})}\\
& =  \max_{(u_i, u_{i+1}) \in U} \frac{d^{S}_{c_2}(u_i, u_{i+1})}{d^{E}_{c_2}(u_i, u_{i+1})}. \qquad\qquad\qquad\qquad\qquad\qquad\qquad\qquad \qquad ~\qed
\end{align*}
\end{proof}
It is clear that the same arguments hold for the directed case.

\section{Intractability}\label{section:Intractability}
We begin by proving that no algorithm is capable of solving the MSp in polynomial time, even if we restrict the considered graphs to be both degree-3 bounded and outerplanar. We do this by showing that there is a family of instances, complying to these restrictions, for which the size of the non-dominated set $\mathcal{Y}_N$ is exponential in the size of the instance, proving the intractability of the MSp.

Consider the following family of instances for $2\leq n \in \mathbb{N}$, of connected, undirected graphs $G=(V,E)$ and edge-weight functions $c_1\colon E \rightarrow \mathbb{Z}$ and $c_2\colon E \rightarrow \mathbb{N}_+$. For every $i \in [n]$, we create vertices $v_i, v_i'$ and $w_i$, and add edges $\{v_i, w_i\}$ with weights $(2^{i}, 2^{i})^\mathsf{T}$, as well as edges  $\{v_i, v_i'\}$ and $\{v_i', w_i\}$ with respective weights of $(0, 2^i)^\mathsf{T}$. Furthermore, we introduce edges $\{w_i,v_{i+1}\}$ with weights $(0,1)^\mathsf{T}$, for $i \in [n-1]$. Finally, we define $v_1\coloneqq s$ and $w_n\coloneqq t$ and add the edge $\{s,t\}$ with weights $(2^{n+1}, 1)^\mathsf{T}$.
An example of this construction can be seen in \Cref{figure:Intractability_deg3}.
Note, that every so constructed graph is degree-$3$ bounded and outerplanar.
\begin{center}
\begin{figure}[b]
\includegraphics[width=\textwidth]{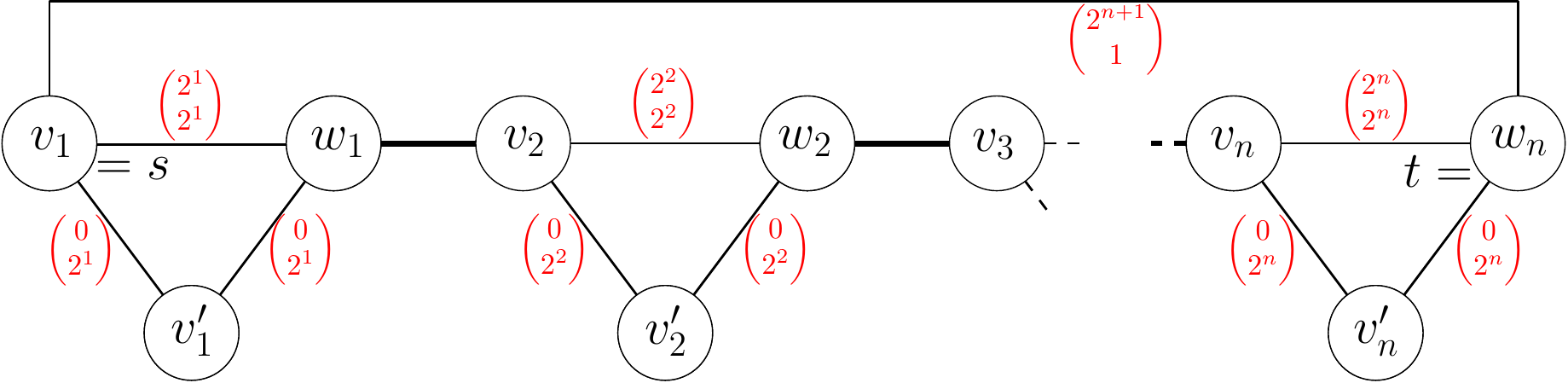}
\caption{The family of instances constructed for \Cref{theorem:MSp_intractable}. Thick edges hold weights $(0,1)^\mathsf{T}$.}
\label{figure:Intractability_deg3}
\end{figure}
\end{center}
Note that the graph $G$ contains at least $2^n$ spanners that do not contain the edge $\{s,t\}$.
With \Cref{observation:spanner_adding_edges_with_c1=0}, we know that for every non-dominated value vector $y$, there is a Pareto-optimal spanner $S$ of $G$ that contains every $e \in E$ with $c_1(e)=0$, with $f(S)=y$.
Define $X$ as the set of all feasible spanners $S$ of $G$, that contain every $e \in E$ with $c_1(e)=0$ and do not contain the edge $\{s,t\}$. 
We now simplify the second objective function $f_2(S)$, for every $ S \in X$, using \Cref{lemma:check_t-spanner}.

\begin{lemma}\label{lemma:intract_f_2}
For all spanners $S \in X$,
$f_2(S)=\max_{u,v \in V} \frac{d^{S}_{c_2}(u,v)}{d^{E}_{c_2}(u,v)}=\frac{d^{S}_{c_2}(s,t)}{d^{E}_{c_2}(s,t)}$.
\end{lemma}
\begin{proof}
Let $S \in X$ be a spanner. With \Cref{lemma:check_t-spanner}, we know that we only have to consider pairs of vertices $u,v \in V$ with $\{u,v\}  \in E$ and $\{u,v\} \notin S$ in order to determine $f_2(S)$. We know that $\{s,t\} \notin S$ holds. Thus,
\[ \frac{d^{S}_{c_2}(s,t)}{d^{E}_{c_2}(s,t)} \geq 
\frac{ \Bigl( \sum_{i=1}^{n}c_2(v_i,w_i)\Bigr) + \left( \sum_{i=1}^{n-1}c_2(w_i,v_{i+1}) \right)}{1}= 2^{n+1}+n-3.\]
The only other pairs of vertices $u,v \in V$ with $\{u,v\} \in E$ for which $\{u,v\} \notin S$ might hold are $v_i,w_i$, for every $i \in [n]$.
We get 
\[\frac{d^{S}_{c_2}(v_i,w_i)}{d^{E}_{c_2}(v_i,w_i)} \leq \frac{c_2(v_i,v_i')+c_2(v_i',w_i)}{c_2(v_i,w_i)} = \frac{2 \cdot 2^i}{2^i}=2 < 2^{n+1}+n-3\leq \frac{d^{S}_{c_2}(s,t)}{d^{E}_{c_2}(s,t)}.\qquad \qed\]
\end{proof}

We now conduct a proof by contradiction to show that two different spanners $S,S' \in X$ do not dominate each other and do not have the same value vector. An expanded proof can be found in \Cref{apendix:intractability}.
\begin{lemma}
For all $S, S' \in  X$ with $S \neq S'$, $S$ and $S'$ do not dominate each other and have different value vectors.
\end{lemma}
\begin{proof}
Let $S, S' \in  X$ be two different spanners and assume $S$ dominates $S'$. Therefore, either $f_1(S)<f_1(S')$ or $f_1(S)=f_1(S')$ holds. We begin by considering the first case. Let $j \in [n]$ be the greatest index at which the shortest s-t-paths in $S$ and $S'$ differ. Since $f_1(S)<f_1(S')$ holds, $S$ must not contain the edge $\{v_j, w_j\}$ while $S'$ has to contain it. Let $P$ be the remaining path that is identical for $S$ and $S'$. Thus, with \Cref{lemma:intract_f_2}, 
\begin{align*}
	f_2(S')&
	 \leq \left( \sum_{i=1}^{j-1} c_2(v_i, v_i')+c_2(v_i', w_i)+c_2(w_i,v_{i+1}) \right) +c_2(v_j, w_j)+c_2(P)\\
	& = \left( \sum_{i=1}^{j-1} 2 \cdot 2^i +1 \right)+2^j +c_2(P) < \left( \sum_{i=1}^{j-1}  2^i +1 \right)+2^j+2^j+c_2(P)\\
	& = \left(\sum_{i=1}^{j-1} c_2(v_i, w_i)+c_2(w_i,v_{i+1}) \right)+c_2(v_j, v_j')+c_2(v_j', w_j)+c_2(P)\\
	& =f_2(S).
\end{align*}
This contradicts the assumed domination.

Let us now consider the second case, in which $f_1(S)=f_1(S')$ holds. Then, in order for $S$ to dominate $S'$, $f_2(S)<f_2(S')$ must hold as well. By design of the $c_1$-edge-weights, we know that in order for $f_1(S)=f_1(S')$ to hold, it is true for every edge $\{v_i, w_i\} \in E$ that $\{v_i, w_i\} \in S \Leftrightarrow \{v_i, w_i\} \in S'$. This claim can be verified by considering that every edge $\{v_i, w_i\}$ has a unique $c_1$-cost that cannot be reproduced by any combination of edges $ e \in E \setminus \{v_i, w_i\}$. Consequently, the shortest s-t-paths in $S$ and $S'$ are exactly the same and therefore $f_2(S)=f_2(S')$ holds, which contradicts the assumed domination. \hfill \qed
\end{proof}

Finally, consider that no $S \in X$ can be dominated by any $\widehat{S} \notin X$.
This is clearly the case, due to the high $c_1$-cost of the edge $\{s,t\}$.
Concluding, we have shown that the set $X$ contains $2^n$ spanners, all of which have different value vectors that are not dominated. Thus, $|\mathcal{Y}_N|\geq |X|=2^n$ and, consequently, \Cref{theorem:MSp_intractable} hold.

Note that an analogous proof can be conducted for the diMSp. The only difference to the undirected case lies in the construction of the family of instances. We turn the family of MSp instances into a family of diMSp instances. For every $i \in [n]$, we replace edges $\{v_i, w_i\}$ with arcs $(v_i, w_i)$, edges  $\{v_i, v_i'\}$ with arcs $(v_i, v_i')$, edges $\{v_i', w_i\}$ with arcs $(v_i', w_i)$. Additionally, we replace edges $\{w_i,v_{i+1}\}$ with arcs $(w_i,v_{i+1})$, for $i \in [n-1]$. Finally, we replace $\{s,t\}$ with $(s,t)$. Every arc holds the same edge-weights as the undirected edge it replaced.

\section{Non-Dominated Set}\label{section:Non-Dominated Set}
In this section, we first consider the output-sensitive complexity of computing the non-dominated set of unweighted (di)MSp instances. Afterwards, we prove that if the (di)MSp can be solved by an output-polynomial algorithm, one can also solve the BUCO problem in output-polynomial time. Therefore, proving that the (di)MSp is \emph{\textbf{BUCO}-hard}.

\begin{observation}\label{observation:trivial_ewf_restricted_YN}
As the trivial edge-weight function only allows linear many values in the range of either of the two objective functions, with \Cref{lemma:poly_restricted_of}, we can infer that the non-dominated set of unweighted MSp instances is only polynomially sized.
\end{observation}
This observation directly infers that the non-dominated set of the (di)MSp cannot be computed in output-polynomial time, as this would enable us to solve the (directed) t-Spanner\textsuperscript{DEC} in polynomial time. Thus, \Cref{theorem:unweighted_MSp_OP} holds.

We now study the output-sensitive complexity of degree bounded outerplanar instances.
For a set of vectors $M$, $\min M$ refers to its non-dominated subset. 
\begin{definition}[Biobjective Unconstrained Optimization (BUCO) Problem \cite{boklerDIS}]
The input are vectors $c^1, c^2 \in \mathbb{N}^n$. A feasible solution is an element of $\{0,1\}^n$. The goal is to find the set of the non-dominated vectors
\[  \mathcal{Y}_N = \min \left\{ \left(  \begin{array}{c} -c^{1^\mathsf{T}} \\ c^{2^\mathsf{T}} \end{array}   \right) x ~ \middle| ~ x \in \{0,1\}^n   \right\}. \]
\end{definition}
The BUCO problem can be interpreted as an unrestricted Knapsack problem.
Without loss of generality, we can assume $c_i^1>0$ for every $i \in [n]$, since any item that does not contribute value is never part of a viable solution. Similarly, we can assume $c_i^2>0$ for every $i \in [n]$.

We prove that the MSp is \textbf{BUCO}-hard by showing that if there is an output-polynomial algorithm $\mathcal{A}$ for the MSp, we could use it to solve the BUCO problem in output-polynomial time. We start by constructing an algorithm that transforms any BUCO instance $I$ into a valid MSp instance $I'$ in polynomial time. Subsequently, we show that the set of non-dominated value vectors of the constructed MSp instance $I'$, that can be found using the algorithm $\mathcal{A}$, can be transformed into the set of non-dominated value vectors of the BUCO instance $I$, using an output-polynomial filter-algorithm.

Let $I$ be an instance of the BUCO problem given by $c^1, c^2 \in \mathbb{N}^n_+$. We construct an instance $I'$ of the MSp with a connected, undirected graph $G$, and edge-weight functions $c_1\colon E \rightarrow \mathbb{Z}$ and $c_2\colon E \rightarrow \mathbb{N}_+$.
Define the constants $C^1 \coloneqq \sum_{i=1}^{n}c_i^1$ and $M \coloneqq C^1 +1$ and construct the graph $G=(V,E)$ in the following way: Create vertices $v_i$, $w_i$ and $v_i'$ for $i \in [n]$ and connect them with edges $\{v_i,w_i\}$ with weights $(0, c_i^2 +2)^\mathsf{T}$, edges $\{v_i,v_i'\}$ with weights $(0,1)^\mathsf{T}$ and edges $\{v_i',w_i\}$ with weights $(c_i^1, 1)^\mathsf{T}$. Furthermore, we add edges $\{w_i,v_{i+1}\}$ with weights $(0,1)^\mathsf{T}$ for $i \in [n-1]$. Define $v_1\coloneqq s$ and $w_n\coloneqq t$. Finally, we add the edge $\{s,t\}$ with weights $(M, 1)^\mathsf{T}$. An example of this construction can be seen in \Cref{figure:OS_BUCO_hard}.

\begin{center}
\begin{figure}[bt]
\includegraphics[width=1\textwidth]{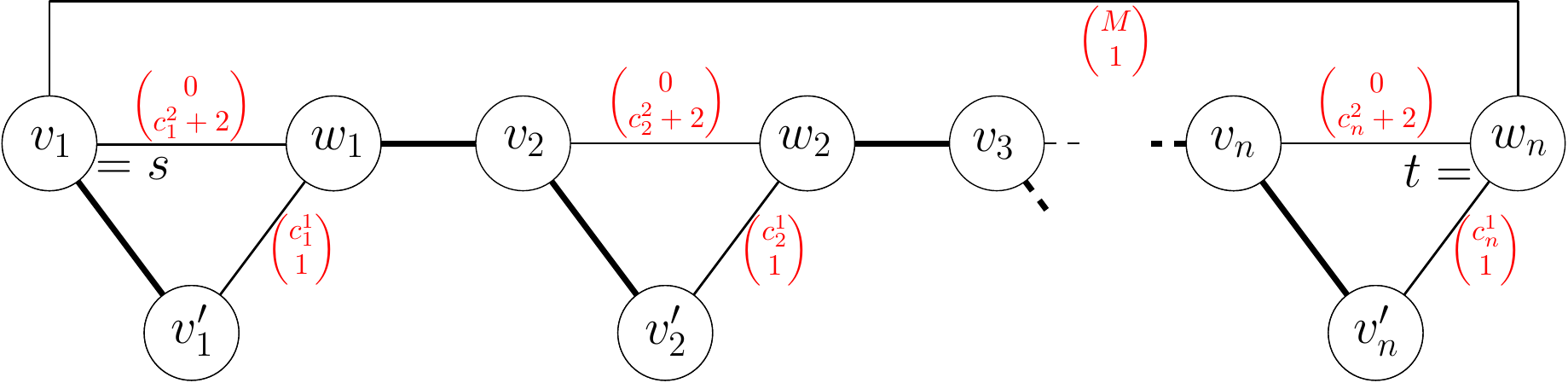}
\caption{Showing the reduction for \Cref{theorem:MSp_BUCO_hard}. Thick edges hold weights $(0,1)^\mathsf{T}$.}
\label{figure:OS_BUCO_hard}
\end{figure}
\end{center}

We observe that every so constructed instance is a valid MSp instance and that all steps can be performed in polynomial time in the size of the instance $I$.
Clearly, the constructed graph $G$ is degree-3 bounded and outerplanar. In order to show that the reduction is correct, we now have to prove that the non-dominated set of the constructed MSp instance can be transformed into the non-dominated set of the initial BUCO instance in output-polynomial time with regard to the instance $I$.

Let $y$ be a value vector of a BUCO instance given by $c^1, c^2 \in \mathbb{N}^n_+$ and let $x \in \{0,1\}^n$ be a solution that is mapped to $y$ with $y= ( -c^{1^\mathsf{T}}x,c^{2^\mathsf{T}}x )^\mathsf{T}$.
There is a spanner $S_x$ of $G$ with the following properties: For every $i \in [n]$, $S_x$ contains edges $\{v_i,w_i\}$ and $\{v_i,v_i'\}$, as well as edges $\{w_i,v_{i+1}\}$, for all $i \in [n-1]$. In addition, if $x_i=0$: $S_x$ contains the edge $\{v_i',w_i\}$. We observe that for every $x \in \{0,1\}^n$ the resulting set of edges $S_x$ is a feasible spanner of $G$ and that none of these spanners contains the edge $\{s,t\}$. Denote the set of all the spanners generated this way as $X$.

With \Cref{observation:spanner_adding_edges_with_c1=0}, we know that for every non-dominated value vector $y$, there is a Pareto-optimal spanner $S$ of $G$ that contains every $e \in E$ with $c_1(e)=0$, with $f(S)=y$.
Thus, clearly analogously to \Cref{lemma:intract_f_2}, $f_2(S)=\max_{u,v \in V} \frac{d^{S}_{c_2}(u,v)}{d^{E}_{c_2}(u,v)}=\frac{d^{S}_{c_2}(s,t)}{d^{E}_{c_2}(s,t)}$ holds for every spanner $S \in X$. We denote the set of edges $\{w_i,v_{i+1}\}$ for $i \in [n-1]$ as $W$.

We now examine how the value vector $y$ of a BUCO solution $x \in \{0,1\}^n$ is connected to the value vector $f(S_x)$ of the corresponding spanner $S_x$.

\begin{lemma}\label{lemma:BUCO-Spanner-value}
For any value vector $y=(-c^{1^\mathsf{T}}x, c^{2^\mathsf{T}}x)^\mathsf{T}$ of a BUCO instance and its associated solution $x \in \{0,1\}^n$, for the constructed corresponding spanner $S_x \in X$,
$f_1(S_x)= C^1 + y^1$ and $f_2(S_x)=y^2 +3n-1$ hold.
\end{lemma}
\begin{proof}
Let $y=(-c^{1^\mathsf{T}}x, c^{2^\mathsf{T}}x)^\mathsf{T}$ be a value vector of a BUCO instance and let $x \in \{0,1\}^n$ be its associated solution. Let $S_x \in X$ be the spanner in the constructed MSp instance that is based on $x$. Consider the two objective functions. 
\begin{align*}
f_1(S_x) &= \left(\sum_{i=1}^{n} c_1(v_i',w_i)\cdot (1-x_i) \right)
= \left(\sum_{i=1}^{n} c_i^1 \cdot (1-x_i)\right)
= C^1 + \sum_{i=1}^{n} -c_i^1  x_i \\
&=C^1 +y^1
\end{align*}
\begin{align*}
f_2(S_x)& = \left( \sum_{i=1}^{n} c_2(v_i,w_i) \cdot x_i+ (c_2(v_i,v_i')+c_2(v_i',w_i))(1-x_i) \right) +c_2(W)\\
& = \left( \sum_{i=1}^{n}(c_i^2+2) x_i  + (1+1) (1-x_i) \right)+(n-1) \cdot 1\\
& = \left( \sum_{i=1}^{n} c_i^2 x_i  +2 \right)+n-1 = y^2 +3n-1 \qquad \qquad \qquad \qquad \qquad \qquad  ~~~ \qed
\end{align*}
\end{proof}

Let us now consider the relationship between spanners $S_x \in X$ and spanners $S$ with $\{s,t\} \in S$.
\begin{lemma}\label{lemma:BUCO_ndom_MSp_ndom}
If $x \in \{0,1\}^n$ is a Pareto-optimal solution for the BUCO instance, its associated spanner $S_x$ is a Pareto-optimal solution for the constructed MSp instance.
\end{lemma}
\begin{proof}
Let $x \in \{0,1\}^n$ be a Pareto-optimal solution for the BUCO instance and let $y=( - c^{1^\mathsf{T}}x, c^{2^\mathsf{T}}x)^\mathsf{T}$ be the corresponding value vector. Let $S_x \in X$ be the spanner generated according to the algorithm defined above. Assume there is a feasible spanner $\widehat{S}$ that dominates $S_x$.
It is clear that due to the $c_1$-cost of the edge $\{s,t\}$, if $\{s,t\} \in \widehat{S}$ holds, $\widehat{S}$ does not dominate $S_x$.
Therefore, we can assume that $\{s,t\} \notin \widehat{S}$ holds. Furthermore, with \Cref{observation:spanner_adding_edges_with_c1=0} and w.l.o.g., we can assume that $\widehat{S}$ contains the edges $\{v_i,v_i'\}$, $\{v_i,w_i\}$ for all $i \in [n]$ and $\{w_i,v_{i+1}\}$ for $i \in [n-1]$. Based on these observations, we can say that $\widehat{S}$ meets the specifications for a spanner that corresponds to a BUCO solution. Let us denote this BUCO solution as $\widehat{x}$ and its value vector as $\widehat{y}$. From now on we refer to $\widehat{S}$ as $S_{\widehat{x}}$. We use \Cref{lemma:BUCO-Spanner-value} to show that if $S_x$ is dominated by $S_{\widehat{x}}$, $\widehat{x}$ dominates $x$. Therefore, causing a contradiction.
Assume $S_{\widehat{x}}$ dominates $S_x$, then $f_1(S_{\widehat{x}}) \leq f_1(S_x)$ and $f_2(S_{\widehat{x}}) \leq f_2(S_x)$ hold. We know that for all spanners $S \in X$ based on a BUCO solution $x'$ with value vector $y_{x'}$, $f_1(S)=C^1 + y^1_{x'} \text{ and } f_2(S_x)=y^2_{x'} +3n-1$ hold. In the first objective function, we therefore get:

\[f_1(S_{\widehat{x}}) \leq f_1(S_x) \Leftrightarrow C^1+ \widehat{y}^1 \leq C^1 + y^1 \Leftrightarrow \widehat{y}^1 \leq y^1.\]
In the second objective function, we get:

\[f_2(S_{\widehat{x}}) \leq f_2(S_x) \Leftrightarrow  \widehat{y}^2+3n-1 \leq y^2 +3n-1 \Leftrightarrow  \widehat{y}^2 \leq y^2.\]
Therefore, either $\widehat{x}$ dominates $x$, which contradicts the assumption that $x$ is Pareto-optimal or $\widehat{x}$ has the same evaluation as $x$, but in that case $S_{\widehat{x}}$ and $S_x$ also have the same evaluation. Hence, the value vector of $S_{\widehat{x}}$ does not dominate the value vector of $S_x$. Therefore, $S_x$ is Pareto-optimal. \hfill \qed
\end{proof}

In order to complete the verification of this reduction, we now prove that there are only polynomially many non-dominated value vectors in the non-dominated set of $I$ that are not based on a Pareto-optimal solution of the BUCO instance.
Consider that the only way a spanner $S$ can divert from the form of a BUCO solution based spanner, without its value vector being dominated, is by containing the edge $\{s,t\}$. Hence, combining \Cref{lemma:poly_restricted_of} with the following Lemma proves the aforementioned claim.

\begin{lemma}\label{lemma:v_l}
For every Pareto-optimal spanner $S$ with $\{s,t\} \in S$ and $l \in [n]$ being the index of the BUCO item with the greatest $c_2$-weight for which $\{v_l',w_l\} \notin S$ holds, $f_2(S)=\frac{d^{S}_{c_2}(v_l',w_l)}{d^{E}_{c_2}(v_l',w_l)}=d^{S}_{c_2}(v_l',w_l)$.
\end{lemma}
\begin{proof}
Let $S$ with $\{s,t\} \in S$ be a Pareto-optimal spanner and let $l \in [n]$ be the index of the BUCO item with the greatest $c_2$-weight for which $\{v_l',w_l\} \notin S$ holds. With \Cref{lemma:check_t-spanner}, and since $\{s,t\} \in S$ and \Cref{observation:spanner_adding_edges_with_c1=0} hold, we know that in order to determine $f_2(S)$, we only have to consider edges $\{v_i',w_i\}$, for all $i \in [n]$. Thus,
\begin{align*}
f_2(S) =\max_{u,v \in V} \frac{d^{S}_{c_2}(u,v)}{d^{E}_{c_2}(u,v)}=\max_{v_i',w_i \in V} \frac{d^{S}_{c_2}(v_i',w_i)}{d^{E}_{c_2}(v_i',w_i)} = \max_{v_i',w_i \in V} \frac{d^{S}_{c_2}(v_i',w_i)}{1} = d^{S}_{c_2}(v_l',w_l).
\end{align*}
Note, that there is an additional edge-case, in which the spanner $S$ contains every $e \in E$. Hence, there are only $n+1$ possible values in the range of the second objective function if the considered spanner contains the edge $\{s,t\}$. \hfill \qed
\end{proof}
With \Cref{lemma:poly_restricted_of}, we infer there are only $n+1$ non-dominated value vectors, in the non-dominated set of the constructed MSp instance, that do not correspond to a Pareto-optimal BUCO solution.
Finally, we describe the algorithm that solves any BUCO instance in output-polynomial time, assuming that there is an output-polynomial algorithm $\mathcal{A}$ capable of solving any MSp instance.

Let $I$ be a BUCO instance, and let $\mathcal{A}$ be an algorithm capable of solving a MSp instance in output-polynomial time. We begin by using the algorithm described above to transform $I$ into the corresponding MSp instance $I'$.
Subsequently, we solve $I'$ using the algorithm $\mathcal{A}$ and receive the set of non-dominated value vectors $\mathcal{Y}_N^{\text{MSp}}$. We know that for every Pareto-optimal solution $x \in \{0,1\}^n$ of the BUCO instance, the constructed MSp instance contains a corresponding Pareto-optimal spanner $S_x$. Therefore, we can assume that for every such $x$, $f(S_x)=(f_1(S_x),f_2(S_x))^\mathsf{T} \in \mathcal{Y}_N^{\text{MSp}}$ holds. Now, we have to filter out all the non-dominated value vectors that do not correspond to a feasible BUCO solution. We do this by inspecting the $y^1$ value for each $y \in  \mathcal{Y}_N^{\text{MSp}}$. If $y^1 \geq M$ holds, then the spanner corresponding to $y$ contains the edge $\{s,t\}$ and consequently is not based on a feasible BUCO solution. If $y^1 \leq C^1 < M$ holds, we transform $y$ according to \Cref{lemma:BUCO-Spanner-value}, so that its values match the ones of the corresponding BUCO solution. We construct $\hat{y}=( y^1-C^1 ,y^2 -3n+1)$ and add it to the set of non-dominated value vectors of the initial BUCO instance $\mathcal{Y}_N^{\text{BUCO}}$. All of these steps are output-polynomial with regard to the BUCO instance $I$ and therefore, the existence of an output-polynomial algorithm for the MSp directly implies the existence of an output-polynomial algorithm for the BUCO problem. Thus, \Cref{theorem:MSp_BUCO_hard} holds.

An analogous reduction can be conducted for the diMSp, by replacing every undirected edge with an arc. This transformation works similar to the one conducted at the end of \Cref{section:Intractability}. Note that the diMSp instances require an additional arc $(v_i',v_i)$ with edge-weights $(0,1)^\mathsf{T}$, for every $i \in [n]$. These arcs ensure that the directed spanners constructed during the reduction are feasible. Observe that the resulting diMSp instances remain outerplanar but are only degree-4 bounded.

\section{Extreme Points}\label{section:Extreme Points}
In this section we define the problem of determining the set of extreme points of a given (di)MSp instance and consider its output-sensitive complexity.

For all $y' \in \mathcal{Y}$, define $W(y')$ as the set of vectors $\lambda \in \mathbb{R}^d_{\geq}, \lambda \neq 0$, so that  $\min_{y \in \mathcal{Y}} \lambda^\mathsf{T} y= \lambda^\mathsf{T} y'$.
A value vector $y' \in \mathcal{Y}$ is called an \emph{extreme point} if there is a $\lambda \in \mathbb{R}^d_{\geq}, \lambda \neq 0$ with
$\lambda \in W(y')$ and $\forall y \in \mathcal{Y}_N \setminus \{y'\}\colon \lambda \notin W(y)$ \cite[Definition 8.7]{ehrgott2005}.
For an instance of the (di)MSp, we define \emph{(di)MSp\textsuperscript{YEx}} to be the problem of computing its set of extreme points $\mathcal{Y}_X$.

We now show that if \textbf{P} $\neq$ \textbf{NP}, even the unweighted MSp\textsuperscript{YEx} cannot be solved in output-polynomial time. With \Cref{theorem:unweighted_MSp_OP}, we know that we cannot compute the entire non-dominated set of MSp instances in output-polynomial time. However, this does not imply the output-sensitive complexity of computing their set of extreme points.

We do this by conducting an indirect reduction of 3SAT.
Consider Cai's reduction of 3SAT to t-Spanner\textsuperscript{DEC}\cite{CAI1994187} for every $2\leq t \in \mathbb{N}$. We turn Cai's constructed 2-Spanner\textsuperscript{DEC} instances into MSp\textsuperscript{YEx} instances and show that iff the initial 3SAT instance is a yes-instance, the yes-witness for the 2-Spanner\textsuperscript{DEC} instance creates an extreme point in the corresponding MSp\textsuperscript{YEx} instance.
In consideration of \Cref{observation:trivial_ewf_restricted_YN}, any output-polynomial algorithm capable of solving unweighted MSp\textsuperscript{YEx} instances in output-polynomial time could solve 3SAT in polynomial time.
We begin with a quick summary of Cai's proof for the special case of $t=2$.

\subsubsection{Revisiting Cai's Proof.}\label{subsubsection:Cais_proof}
Cai transforms 3SAT to 2-Spanner\textsuperscript{DEC}. Given an instance $I=(U, C)$ of 3SAT consisting of a set $U$ of $n$ distinct variables and a collection $C$ of $m$ 3-element clauses over $U$. They construct a 2-Spanner\textsuperscript{DEC} instance $\hat{I}$, with a graph $G=(V,E)$ and a positive integer $K \in \mathbb{N}_+$, such that $G$ contains a 2-spanner with at most $K$ edges if and only if $C$ is satisfiable.
They define a \emph{2-path} as a path with 2 edges. One can force an edge to be in any minimum 2-spanner of a graph by the addition of two distinct 2-paths between the two ends of the edge \cite[Lemma 3]{CAI1994187}. This operation is called \emph{forcing an edge}.
Such an edge is called a \emph{forced edge} and the two 2-paths are called \emph{forcing paths}.

They construct the \emph{truth-setting component} $T$ as follows: They take five vertices $z$, the literal vertices $x$ and $\bar{x}$, and the y-type vertices $y$ and $y'$; and join $z$ to each of the remaining four vertices by an edge. Finally, they add forced edges $\{x,\bar{x}\}, \{x,y\}, \{x,y'\}, \{\bar{x},y\}, \{\bar{x},y'\}$.

They assign each variable $u_i \in U$, $i \in [n]$ a distinct copy $T_i$ of $T$ and identify all vertices $z_i$ into a single vertex $z$ to form a subgraph $T'$ of $G$. To finish the construction of $G$, they create a new vertex $v_i$ for each clause $c_i \in C$, $i \in [m]$, join it to vertex $z$ with an edge and add a forced edge between $v_i$ and each of the three literal vertices in $T'$ corresponding to the three literals of $c_i$.

They finish the construction of the 2-Spanner\textsuperscript{DEC} instance, by setting $K=16n+9m$. 
For the constructed graph $G$, it holds that any minimum 2-spanner $S$ of $G$ contains at least $K$ edges. Furthermore, if $S$ contains exactly $K$ edges, then for each $T_i$, $i \in [n]$ exactly one of the two literal edges $\{z,x_i\}$ and $\{z,\bar{x_i}\}$ belongs in $S$ \cite[Lemma 4]{CAI1994187}.
Examples of the described constructions can be seen in Figures \ref{apendix:figure:WSS_Cai_T2} and \ref{apendix:figure:WSS_Cai_complete} in the appendix.


Now, suppose that $C$ is satisfiable and let $\phi$ be a satisfying truth assignment for $C$.
They construct a yes-witness-spanner $S_w$ as follows: put every forced edge in $S_w$. For each forcing path, put one of the two edges in $S_w$. Finally, for each variable $u_i \in U$, if $u$ is \enquote{true} under $\phi$ then put edge $\{z,x_i\}$, in $S_w$ else put edge $\{z,\bar{x_i}\}$ in $S_w$.
The complete proof that this is a correct reduction goes beyond the scope of this paper and can be found in the original paper \cite{CAI1994187}. Instead, let us now construct an equivalent MSp\textsuperscript{YEx} instance and prove that if the initial 2-Spanner\textsuperscript{DEC} instance is a yes-instance, the value vector of the yes-witness-spanner is an extreme point.

\subsubsection{The Associated MSp\textsuperscript{YEx} Instance.}
First, let $I$ be a 3SAT instance and let $\hat{I}=(G=(V,E),K)$ be the associated 2-Spanner\textsuperscript{DEC} instance, constructed according to the algorithm described in Cai's proof. We turn $\hat{I}$ into an instance $I'=(G,c)$ of the unweighted MSp\textsuperscript{YEx} by copying $G$ and adding the trivial edge-weight function $c\colon E\rightarrow \{(1,1)^\mathsf{T}\}$, $c(e)=(1,1)^\mathsf{T}$ for all $e \in E$. Clearly, this can be done in polynomial time. Note, that $I'$ is a valid, unweighted MSp\textsuperscript{YEx} instance. Now, let $\hat{I}$ be a yes-instance and let $S_w$ be the yes-witness-spanner. It is clear that the same spanner exists in $I'$ and that 
\[f_1(S_w)=\sum_{e \in S_w}c_1(e)=\sum_{e \in S_w} 1 = |S_w|=K
\text{ and }
f_2(S_w)=\max_{u,v \in V} \frac{d^{S_w}_{c_2}(u,v)}{d^{E}_{c_2}(u,v)}=2\]
hold.
We now show that the value vector of the yes-witness-spanner $f(S_w)$ is an extreme point by finding the non-dominated value vectors $\mathcal{Y}_N$ and showing that there is a $\lambda_{w} \in \mathbb{R}^d_{\geq}, \lambda_{w} \neq 0$, so that $\lambda_{w} \in W(f(S_w)) \text{ and } \forall y \in \mathcal{Y}_N \setminus \{f(S_w)\}\colon \lambda_{w} \notin W(y)$ hold.

Note that $f(S_w)$ dominates or is equal to every value vector $y \in \mathcal{Y}$ with $y^1\geq K$ and $y^2 \geq 2$. Let us first find the non-dominated value vectors $y \in \mathcal{Y}_N$ with $y^2<2$. We begin by constructing the spanner $S_1 \subseteq E$ that contains the fewest edges for which $f_2(S_1) < f_2(S_w)=2$ holds.
Consider $S_1=E$ and observe that if we remove any edge $e \in E$ from $S_1$, $f_2(S_1 \setminus \{e\})=2$ holds. Hence, $S_1=E$ and we conclude that for every feasible spanner $S \subseteq E$ with $K<f_1(S)<|E|$, $f(S_w)$ dominates $f(S)$. Consider the amount of edges in $G$. For every truth-setting component $T_i$, $i \in [n]$ there are $20$ forcing edges, $5$ forced edges and $4$ edges connecting to the vertex $z$. Furthermore, for each clause $c_i \in C$, $i \in [m]$ there are $12$ forcing edges, $3$ forced edges and one edge connecting to $z$. Hence,
\[f_1(S_1)=|E|=29n + 16m \text{ and } f_2(S_1)=1.\]

We now construct a hypothetical value vector $y_h$ that either dominates or is equal to every value vector $f(S)$ of feasible spanners $S$ with $f_1(S)<K$. We begin by constructing $y^1_h$. 
Consider the number of vertices in $G$.
For every truth-setting component $T_i$, $i \in [n]$ there are $10$ vertices that are part of forcing paths and $4$ vertices $x_i,\bar{x_i},y_i,y_i'$. Furthermore, for each clause $c_i \in C$, $i\in [m]$ there is one vertex $v_i$ and $6$ vertices that are part of forcing paths. Finally, there is the vertex $z$. Thus, there are $14n+7m+1$ vertices in the $G$.
Therefore, for every feasible spanner $S$ of $G$, $f_1(S)=|S| \geq 14n+7m$. Hence, we set $y^1_h= 14n+7m$. Let us now focus on $y^2_h$. We know that there are no spanners $S$ that contain fewer edges than $S_w$ for which $f_2(S)\leq 2$ hold. Hence, we set $y^2_h=3$. Thus,
\[ y_h^1= 14n+7m \text{ and } y_h^2=3.\]
Clearly, for every value vector $f(S)$ of a feasible spanner $S$ with $f_1(S)<K$, $y^1_h\leq f_1(S)$ and $y^2_h\leq f_2(S)$ hold. Hence, $y_h$ either dominates or is equal to every such value vector and thus, for every $\lambda \in \mathbb{R}^d_{\geq}, \lambda \neq 0$, $\lambda^\mathsf{T} y_h \leq \lambda^\mathsf{T} \cdot f(S)$ holds.
An expanded proof for the following Lemma can be found in \Cref{appendix:Extreme_Points}.
\begin{lemma}
The value vector of the yes-witness-spanner is an extreme point.
\end{lemma}
\begin{proof}
Consider the vector $\lambda_{w}=(2,15n+9m)^\mathsf{T} \in \mathbb{R}^d_{\geq}, \lambda_{w} \neq 0$. We show that $\lambda_{w} \in W(f(S_w))$ holds and that for all $y \in \mathcal{Y}_N \setminus \{f(S_w)\}\colon \lambda_{w} \notin W(y)$.

\begin{gather*}
    \lambda_{w}^\mathsf{T} \cdot f(S_1) = 73n+41m = \lambda_{w}^\mathsf{T} y_h\\
    \lambda_{w}^\mathsf{T} \cdot f(S_w) = 62n+36m < 73n+41m = \lambda_{w}^\mathsf{T} \cdot f(S_1) = \lambda_{w}^\mathsf{T} y_h, 
\end{gather*}
for $n,m >0$. Hence, $\lambda_{w} \in W(f(S_w))$ and for all $y \in \mathcal{Y}_N \setminus \{f(S_w)\} \colon \lambda_{w} \notin W(y)$. \hfill \qed
\end{proof}
A sketch of the relevant value vectors can be seen in \Cref{apendix:figure:WSS_Punkte} in the appendix.

Finally, let us consider the entire reduction.
Suppose there is an algorithm $\mathcal{A}$ capable of solving unweighted MSp\textsuperscript{YEx} instances in output-polynomial time.
We know that iff the initial 3SAT instance $I$ is a yes-instance, the 2-Spanner\textsuperscript{DEC} instance $\hat{I}$ constructed according to Cai's proof is a yes-instance too. Therefore, the yes-witness-spanner $S_w$ exists in $\hat{I}$ and thus, also in the associated MSp\textsuperscript{YEx} instance $I'$.
Since the value vector $f(S_w)$ of $S_w$ is an extreme point, it is therefore part of the output $\mathcal{Y}_X$ of algorithm $\mathcal{A}$, when applied to $I'$.
In consideration of $\mathcal{Y}_X \subseteq \mathcal{Y}_N$ and \Cref{lemma:poly_restricted_of} we infer that solving $I'$ with $\mathcal{A}$ and checking whether $f(S_w) \in \mathcal{Y}_X$ holds is possible in polynomial time in the size of $I$.
In conclusion, if $\mathcal{A}$ existed, we could solve 3SAT in polynomial time. Thus, \Cref{theorem:MSP_YEx_notin_OP} holds.


Similarly, based on Cai's reduction of 3SAT to directed 2-Spanner\textsuperscript{DEC} \cite[Section 3]{CAI1994187}, we can show the same results for the diMSp\textsuperscript{YEx}. The only difference to the undirected case lies in the construction of the directed 2-Spanner\textsuperscript{DEC} instance by Cai.
These differences in turn cause slightly different values in the objective functions of the considered spanners in the diMSp\textsuperscript{YEx} instance that we construct analogously to the undirected case. It is easy to see that these differences have no influence on the validity of the statement. Due to the analogy of the proofs, we leave the details to the reader.

\section{Conclusion}\label{section:Conclusion}
What remains open is the output-sensitive complexity of computing the set of extreme points for degree-3 bounded outerplanar instances, as it is currently unknown whether there is a stretch factor $t$, such that the t-Spanner\textsuperscript{DEC} problem is \textbf{NP}-hard under these restrictions. Thus, such a prove requires a different approach than the one used in \Cref{section:Extreme Points}.

Future work might include the development of approximation techniques for the MSp and related problems, as well as investigating what existing approaches can be applied to them.

\printbibliography

\newpage

\begin{appendix}
\section{Intractabilty}\label{apendix:intractability}
\begin{customlemma}{4}
For all $S, S' \in  X$ with $S \neq S'$, $S$ and $S'$ do not dominate each other and have different value vectors.
\end{customlemma}
\begin{proof}
Let $S, S' \in  X$ be two different spanners and assume $S$ dominates $S'$. Therefore, either $f_1(S)<f_1(S')$ or $f_1(S)=f_1(S')$ holds. We begin by considering the first case. Let $j \in [n]$ be the greatest index at which the shortest s-t-paths in $S$ and $S'$ differ. Since $f_1(S)<f_1(S')$ holds, $S$ must not contain the edge $\{v_j, w_j\}$ while $S'$ has to contain it. Let $P$ be the remaining path that is identical for $S$ and $S'$. Thus,
\begin{align*}
	f_2(S')&= d^{S'}_{c_2}(s,t)\\
	& \leq \left( \sum_{i=1}^{j-1} c_2(v_i, v_i')+c_2(v_i', w_i)+c_2(w_i,v_{i+1}) \right) +c_2(v_j, w_j)+c_2(P)\\
	& = \left( \sum_{i=1}^{j-1} 2 \cdot 2^i +1 \right)+2^j +c_2(P) = \left( \sum_{i=2}^{j} 2^i +1 \right)+2^j +c_2(P)\\
	& < \left( \sum_{i=1}^{j-1}  2^i +1 \right)+2^j+2^j+c_2(P)\\
	& = \left(\sum_{i=1}^{j-1} c_2(v_i, w_i)+c_2(w_i,v_{i+1}) \right)+c_2(v_j, v_j')+c_2(v_j', w_j)+c_2(P)\\
	& \leq d^{S}_{c_2}(s,t)=f_2(S).
\end{align*}
This contradicts the assumed domination.

Let us now consider the second case, in which $f_1(S)=f_1(S')$ holds. Then, in order for $S$ to dominate $S'$, $f_2(S)<f_2(S')$ must hold as well. By design of the $c_1$-edge-weights, we know that in order for $f_1(S)=f_1(S')$ to hold, it is true for every edge $\{v_i, w_i\} \in E$ that $\{v_i, w_i\} \in S \Leftrightarrow \{v_i, w_i\} \in S'$. This claim can be verified by considering that every edge $\{v_i, w_i\}$ has a unique $c_1$-cost that cannot be reproduced by any combination of edges $ e \in E \setminus \{v_i, w_i\}$. Consequently, the shortest s-t-paths in $S$ and $S'$ are exactly the same and therefore $f_2(S)=f_2(S')$ holds, which contradicts the assumed domination. \hfill \qed
\end{proof}

\section{Extreme Points}\label{appendix:Extreme_Points}
\begin{customlemma}{8}
The value vector of the yes-witness-spanner $f(S_w)$ is an extreme point.
\end{customlemma}
\begin{proof}
Consider the vector $\lambda_{w}=(2,15n+9m)^\mathsf{T} \in \mathbb{R}^d_{\geq}, \lambda_{w} \neq 0$. We show that $\lambda_{w} \in W(f(S_w))$ holds and that for all $y \in \mathcal{Y}_N \setminus \{f(S_w)\}\colon \lambda_{w} \notin W(y)$. Begin by considering $\lambda_{w}^\mathsf{T} \cdot f(S_1)$ and $\lambda_{w}^\mathsf{T} y_h$.

\begin{align*}
\lambda_{w}^\mathsf{T} \cdot f(S_1) & = \begin{pmatrix}  2 & 15n+9m \end{pmatrix} \cdot \begin{pmatrix}29n + 16m \\ 1 \end{pmatrix}
= 73n+41m\\
& = \begin{pmatrix}  2 & 15n+9m \end{pmatrix} \cdot \begin{pmatrix}14n+7m \\ 3\end{pmatrix}  = \lambda_{w}^\mathsf{T} y_h
\end{align*}

Now, consider $\lambda_{w}^\mathsf{T} \cdot f(S_w)$. We get
\begin{align*}
\lambda_{w}^\mathsf{T} \cdot f(S_w) & =  \begin{pmatrix}  2 & 15n+9m \end{pmatrix} \cdot \begin{pmatrix} K \\ 2\end{pmatrix} =  \begin{pmatrix}  2 & 15n+9m \end{pmatrix} \cdot \begin{pmatrix} 16n+9m \\ 2\end{pmatrix} \\
& = 62n+36m < 73n+41m = \lambda_{w}^\mathsf{T} \cdot f(S_1) = \lambda_{w}^\mathsf{T} y_h,
\end{align*}
for $n,m >0$. Hence, $\lambda_{w} \in W(f(S_w))$ and for all $y \in \mathcal{Y}_N \setminus \{f(S_w)\} \colon \lambda_{w} \notin W(y)$. \hfill \qed
\end{proof}

\begin{center}
\begin{figure}
\centering
\includegraphics[scale=0.4]{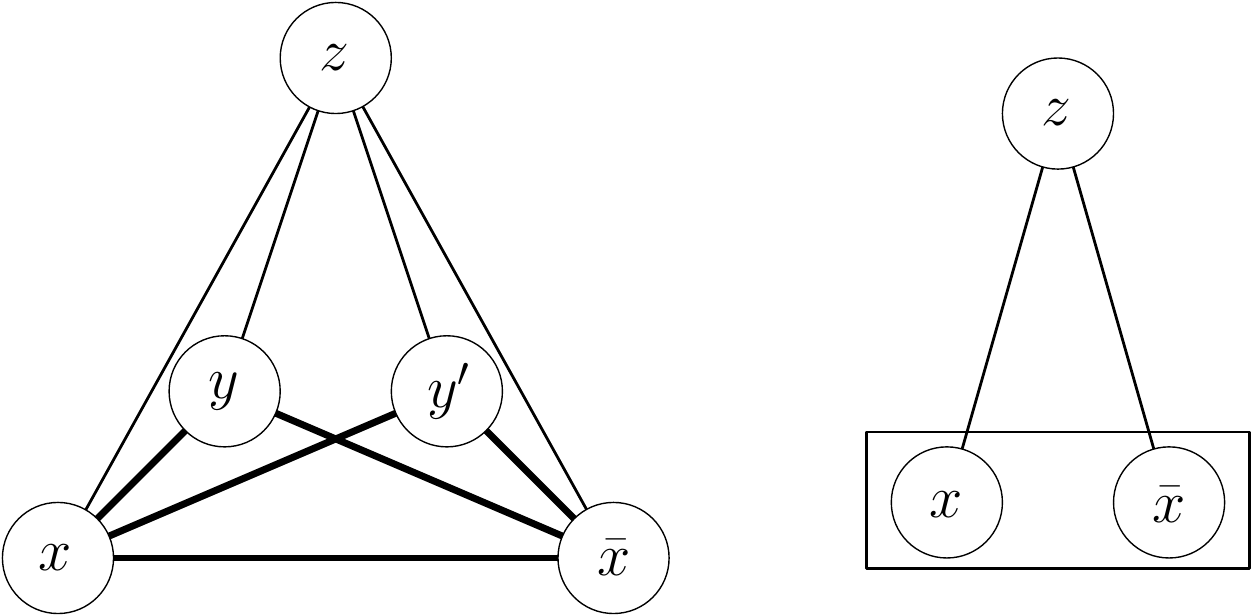}
\caption{A truth-setting component $T$ and its symbolic representation. Thick edges indicate forced edges. For clarity, forcing paths have been omitted from the figure.}
\label{apendix:figure:WSS_Cai_T2}
\end{figure}
\end{center}

\begin{center}
\begin{figure}
\centering
\includegraphics[scale=0.5]{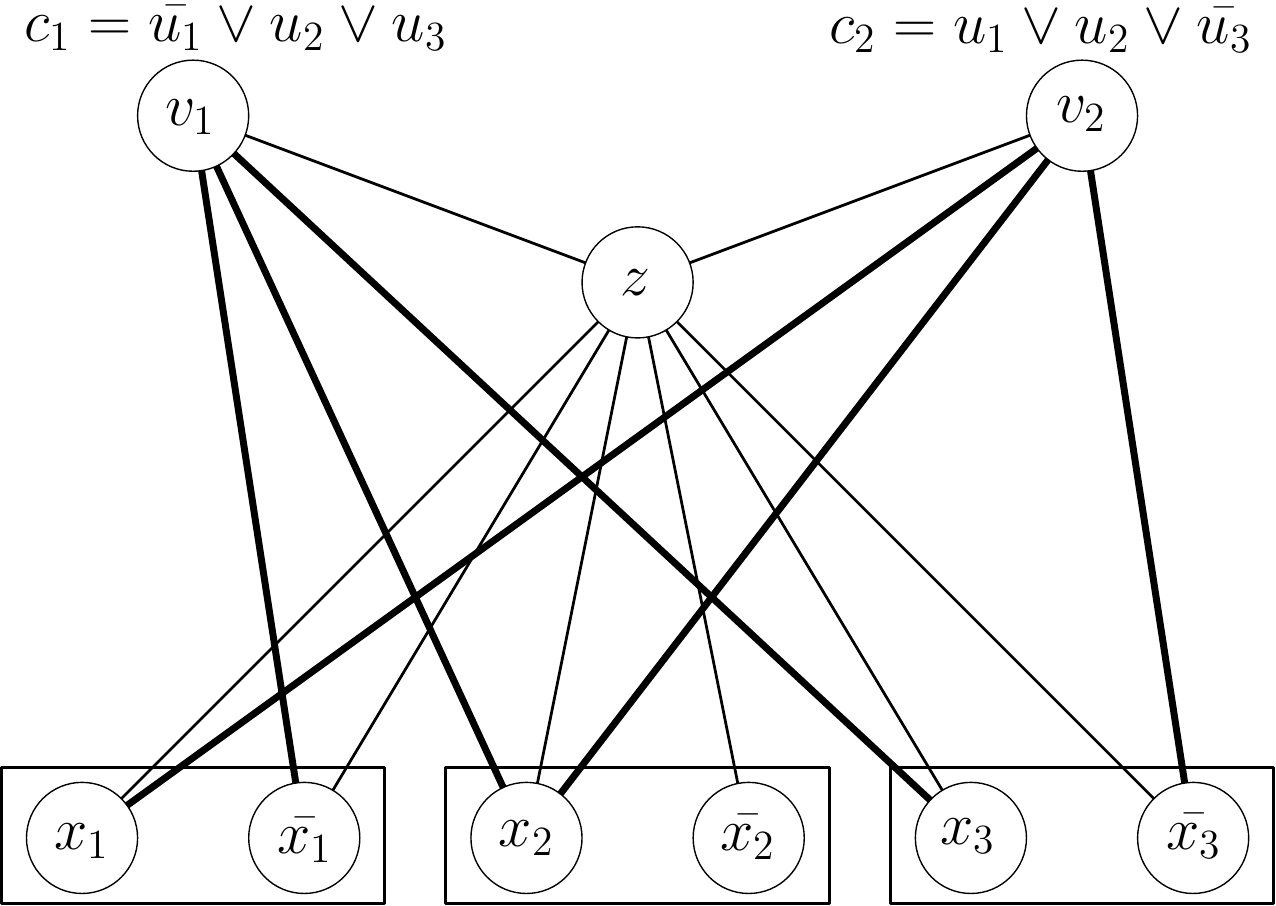}
\caption{The graph $G$ for $C=\{ \{\bar{u_1}, u_2, u_3\}, \{u_1, u_2, \bar{u_3} \} \}$. Thick edges indicate forced edges. For clarity, forcing paths have been omitted from the figure.}
\label{apendix:figure:WSS_Cai_complete}
\end{figure}
\end{center}

\begin{center}
\begin{figure}
\centering
\includegraphics[scale=0.5]{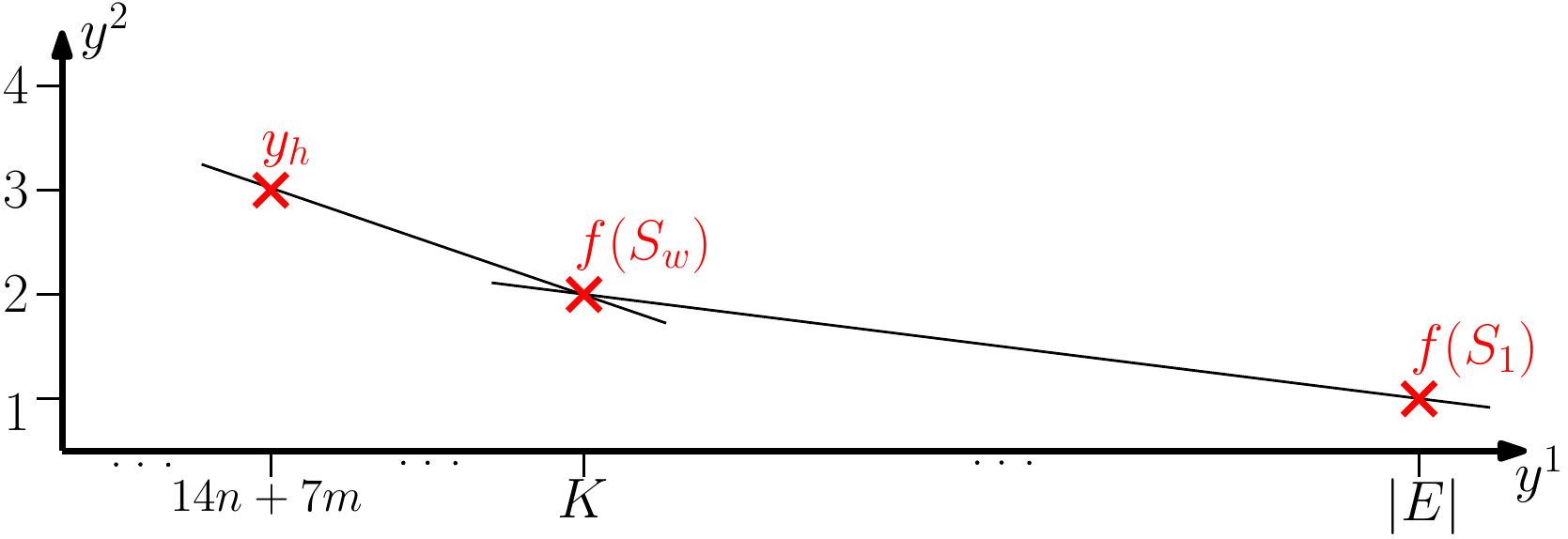}
\caption{The value vectors $y_h,f(S_w), f(S_1)$ and slopes visualized.}
\label{apendix:figure:WSS_Punkte}
\end{figure}
\end{center}
\end{appendix}
\end{document}